\pdfoutput=1
\documentclass[11pt]{article}
\usepackage[letterpaper, left=1in, right=1in, top=1in,bottom=1in]{geometry}

\usepackage[ruled]{algorithm2e} 

\SetAlFnt{\small}
\SetAlCapFnt{\small}
\SetAlCapNameFnt{\small}
\SetAlCapHSkip{0pt}
\IncMargin{-\parindent}

\usepackage{color}              
\usepackage[suppress]{color-edits}
\addauthor{et}{blue}
\addauthor{pj}{green}
\usepackage{hyperref}
\usepackage{graphicx}
\usepackage{amsmath,amssymb,amsthm}
\usepackage{tikz}
\usetikzlibrary{chains,arrows}

\theoremstyle{plain}

\newtheorem{theorem}{Theorem}[section]
\newtheorem{lemma}[theorem]{Lemma}
\newtheorem{corollary}[theorem]{Corollary}
\theoremstyle{definition}

\def\squareforqed{\hbox{\rule{2.5mm}{2.5mm}}}

\def\QED{\ifmmode\squareforqed 
  \else{\nobreak\hfil   
    \penalty50                 
    \hskip1em                  
    \null                      
    \nobreak                   
    \hfil                      
    \squareforqed              
    \parfillskip=0pt           
    \finalhyphendemerits=0     
    \endgraf}                  
  \fi}

\def\blksquare{\rule{2mm}{2mm}}
\def\qedsymbol{\blksquare}
\newcommand{\bg}[1]{\medskip\noindent{\bf #1}}
\newcommand{\ed}{{\hfill\qedsymbol}\medskip}

\newtheorem{thm}{Theorem}

\newcommand{\SubmissionOmit}[1]{} 

\newcommand{\fullversion}[1]{}

\usepackage[ruled]{algorithm2e}

\usepackage{amsmath,amsfonts,amssymb,bbm}

\usepackage{dsfont}
\usepackage{ifthen}
\usepackage{commath}
\usepackage{graphicx}
\usepackage{enumerate}
\usepackage{color}
\usepackage{tikz}
\usepackage{hyperref}
\usepackage[capitalize]{cleveref}

\usetikzlibrary{calc,math,graphs,positioning,quotes,arrows.meta} 

\tikzstyle{marketstyle}=[x=3cm,y=2cm,node font=\small]%
\tikzstyle{every edge quotes}=[auto=false,fill=white,opacity=.95,text opacity=1,pos=.25,inner sep=1pt]%
\tikzstyle{every node}+=[inner sep=2.5pt]%

\SetArgSty{textrm}  
\SetAlFnt{\small}
\SetAlCapFnt{\small}
\SetAlCapNameFnt{\small}
\SetAlCapHSkip{0pt}
\IncMargin{-\parindent}

\DeclareMathAlphabet{\mathbfit}{OML}{cmm}{b}{it}

\newcommand{\suchthat}{\;\ifnum\currentgrouptype=16 \middle\fi|\;}

\newcommand{\myref}[2]{\hyperref[#2]{$#1$\ref*{#2}}}

\newcommand{\AutoAdjust}[3]{\mathchoice{ \left #1 #2  \right #3}{#1 #2 #3}{#1 #2 #3}{#1 #2 #3} }

\newcommand{\InBrackets}[1]{\AutoAdjust{[}{#1}{]}}

\newcommand{\RangeAB}[2]{\AutoAdjust{[}{{#1},{#2}}{]}}

\newcommand{\RangeN}[1]{\InBrackets{#1}}

\newcommand{\Reals}{\mathbb{R}}

\newcommand{\PosReals}{\Reals_{\ge 0}}
\newcommand{\SPosReals}{\Reals_{+}}

\newcommand{\PosInts}{\mathbb{N}}

\newcommand{\SpaceStyle}[1]{{\mathrm #1}}
\newcommand{\SpacesStyle}[1]{{\mathbf {#1}}}

\newcommand{\N}{n}

\newcommand{\NumAgents}{n}
\newcommand{\NumItems}{m}
\newcommand{\ItemCount}{C}
\newcommand{\AgentID}[1]{\ensuremath{\operatorname{A}_{#1}}}
\newcommand{\ItemID}[1]{\ensuremath{\operatorname{I}_{#1}}}
\newcommand{\RSD}{RSD}
\newcommand{\Alloc}{x}
\newcommand{\Allocs}{{\mathbf x}}
\newcommand{\Allocss}{{\mathbf x}}
\newcommand{\AllocsY}{{\mathbf y}}
\newcommand{\AllocssY}{{\mathbf y}}
\newcommand{\AgentDim}{m}
\newcommand{\AllocSpace}{\SpaceStyle{S}}
\newcommand{\AllocSpaces}{\SpacesStyle{S}}
\newcommand{\ConsFun}{h}
\newcommand{\NumCons}{d}

\newcommand{\Price}{p}
\newcommand{\Prices}{{\mathbf p}}
\newcommand{\Val}{v}
\newcommand{\Util}{u}
\newcommand{\Utils}{{\mathbf u}}
\newcommand{\DualUtil}{g}
\newcommand{\DemandSup}{D}
\newcommand{\Response}{R}

\crefname{property}{Property}{Properties}

\crefname{program}{Program}{Programs}

\begin{document}


\title{Computing Equilibrium in Matching Markets}




\author{%
	Saeed Alaei \thanks{Google Inc. Email: \texttt{alaei@google.com}.}
	\and
	Pooya Jalaly \thanks{Department of Computer Science, Cornell University, Gates Hall, Ithaca, NY 14853, USA, Email: \texttt{jalaly@cs.cornell.edu}. Supported in part by NSF grant CCF-1563714, and a Google Research Grant.}
	\and
	\'{E}va Tardos\thanks{Department of Computer Science, Cornell University, Gates Hall, Ithaca, NY 14853, USA, Email: \texttt{eva@cs.cornell.edu}. Work supported in part by NSF grant  CCF-1563714, ONR grant N00014-08-1-0031, and a Google Research Grant.}
}

\maketitle
\begin{abstract}
Market equilibria of matching markets offer an intuitive and fair solution for matching problems without money with agents who have preferences over the items. Such a matching market can be viewed as a variation of Fisher market, albeit with rather peculiar preferences of agents. These preferences can be described by piece-wise linear concave (PLC) functions, which however, are not separable (due to each agent only asking for one item), are not monotone, and do not satisfy the gross substitute property-- increase in price of an item can result in increased demand for the item. Devanur and Kannan in FOCS 08 showed that market clearing prices can be found in polynomial time in markets with fixed number of items and general PLC preferences. They also consider Fischer markets with fixed number of agents (instead of fixed number of items), and give a polynomial time algorithm for this case if preferences are separable functions of the items, in addition to being PLC functions.
 
Our main result is a polynomial time algorithm for finding market clearing prices in matching markets with fixed number of different agent preferences, despite that the utility corresponding to matching markets is not separable. We also give a simpler algorithm for the case of matching markets with fixed number of different items.
\end{abstract}

\newpage




\section{Introduction}
We consider the problem of matching without money with $\NumAgents$ agents who
have 
preferences over $\NumItems$ items. \etedit{This problem models a range of situations} from assigning students to schools, applicants to jobs, or people to committees. We call such an assignment problem a matching problem, if 
all agents are required to get a fixed number of items. An intuitive application is the school choice problem. Students have preferences over schools,
and each student needs to get assigned to exactly one school. In this paper we will consider computing the fair randomized solution to this problem proposed by \cite{HZ79} based on market equilibria.

We model the preferences of agents with the value of agent $i$ for being assigned to item $j$ is $v_{ij}$. Using values allows agents to
express the intensity of their preferences. An important property of an allocation is its efficiency. Since agents utilities are meaningless to compare (without money, there is no natural unit to express utility), the best we can hope for is a Pareto-efficient allocation. An allocation is Pareto inefficient if there is an alternate allocation where no agent is worse off, and at least one agent has improved utility. We will consider fractional or randomized allocation. The value of an agent $i$ for the fractional allocation $x_{ij}$ is $\sum_j v_{ij}x_{ij}$, if it obeys the matching constraint. This is the agent's expected value, if $x_{ij}$ is the probability of assigning item $j$ to agent $i$.


Market equilibria offer an intuitive, fair, and Pareto-efficient solution for problems of allocations of resources to agents who have their own (incomparable) preferences over the items in systems with no money. This was proposed by \cite{HZ79} in the context of matching markets, and by \cite{Dolev:2012} (see also \cite{Gutman-Nisan:2012} and \cite{Martinez:2015}), in the context of allocation of resources in systems. The idea is to endow each agent with equal resource: a unit of (artificial) money. A set of prices $\Prices$ for the items is market clearing, if there is a fractional allocation $\Allocss$ of items to agents such that the following conditions hold (i) each item is allocated at most once, (ii) each agents is allocated her favorite set of items subject to the budget constraint\footnote{Note that agents have no use for the (artificial) money and are simply optimizing their allocated item, subject to their budget.} that $\sum_j \Alloc_{ij}p_j \le 1$, and (iii) the market clears, meaning that all items not fully allocated have price 0. \cite{HZ79} showed that such market equilibria is guaranteed to exists, see also Appendix \ref{sec:existence}. We view the resulting fractional (or randomized) allocation $\Allocss$ as a fair solution to the allocation problem without money, which is also clearly Pareto-efficient and envy-free (no agent prefers the allocation of another agent).\footnote{An alternate way to arrive to the same solution concept is to assign each agent an equal share of each resource, and then look for an equilibrium of the resulting exchange market. To see that this results in an identical outcome, we can think of each agents trade, as a two step-process, where he first sells all his allocated share on the
market prices, and then uses the resulting money to buy his optimal allocation. } We are concerned with computing this solution efficiently.

\paragraph{Computing Market Equilibria and the Odd Demand Structure of Matching Markets.}
Market equilibrium problems where demands satisfy the \emph{gross substitute} condition
are well understood \cite{Bruno:2007}, and can be computed efficiently. The demand structure of our matching problem does not satisfy the gross substitutability condition, which requires that decreasing the price of an item (while keeping all other prices fixed) should never decrease the demand for that item. We show an example in Appendix \ref{ap:examples} that decreased price can cause decreased demand in a matching markets. It is not hard to gain intuition for the phenomena: with the decreased price the agent could get her old allocation, and would have money leftover. In other markets, money can be used to buy additional items. However, in a matching market additional item makes no sense, and instead, the agent may want to exchange his share of a cheaper and less favorable item (possibly the item whose price decreased) for share of a more valuable expensive item.

\cite{Devanur:2008} gave an algorithm to compute market equilibria in markets with a fixed number of items, where agents have piece-wise linear concave (PLC) utility functions, despite the fact that PLC utility functions can give rise to demand not satisfying the gross substitute condition. They also gave a polynomial time algorithm to compute market equilibria in markets with a fixed number of agents, where agents have piece-wise linear concave and \emph{separable} utility functions. They leave as an open problem to give a polynomial time algorithm to compute market equilibria in markets with a fixed number of agents and general PLC utilities that are not separable. We will show in Section \ref{sec:prelim} that demand structure of the matching market can be modeled by a piece-wise linear concave (PLC) utility function, which however, is neither separable nor monotone. This allows us to use the algorithm of \cite{Devanur:2008} to find a market equilibrium is the number of goods is fixed, but leaves open the question whether market equilibrium can be found in polynomial time if the number of different agents is finite instead.

\paragraph{Our Results.}
We give a polynomial time exact algorithm for finding market equilibria of matching markets with a fixed number of agents, extending the work of \cite{Devanur:2008} to the case of matching markets with a fixed number of agents, despite the fact that utilities describing matching markets are not separable. Our algorithm in Section \ref{Sec:FixedAgents} is based on the structural Theorem \ref{Thm:FixedAgentsEqChar}, and explores a polynomial number of possible player utility values and allocation structures, and finds a market equilibrium in polynomial time when the number of agents is fixed. The algorithm also extend to the case when there are only a finite number of different agent utility types.

In case of large number of items and finite number of agents, when each individual item is insignificant, our allocation can be used for finding an approximately optimal integer solution. We achieve this by showing in Lemma \ref{Lem:ReBundling}  that we find an allocation in which the number of items which are shared by the agents is $O(\NumAgents^2)$, which is constant when the number of agents is constant.

In Appendix \ref{Sec:FixedGoods}, we consider the problem with a fixed number of goods. In this case, the algorithm of \cite{Devanur:2008} can find market equilibria in polynomial time. We give a simpler algorithm which is tailored for matching markets. With $m$ different goods and $n$ agents, our algorithm enumerates a polynomial number of different set of prices and allocation structures for the equilibrium.

%

\paragraph{Related Work.}
The problem of fairly allocating items to unit demand agents without money has been
studied extensively in both Economics and Computer Science literature.  Perhaps
the most well known solution to this problem is the \emph{random serial dictatorship (RSD)} \cite{AS98} --- also known as \emph{random priority (RP)} --- in which agents are served
sequentially according to a random permutation, and each agent in turn receives
her most preferred item among the remaining ones. Clearly, serial dictatorship is Pareto efficient, and as a result,  \RSD{} is ex post Pareto-efficient, i.e., Pareto-efficient given the order used. However, the expected allocation of \RSD{} may not be Pareto-efficient, i.e., its  not interim Pareto-efficient. In Appendix \ref{ap:examples} we give an example where the expected allocation of \RSD{} can be Pareto improved just using the order of player preferences, showing that \RSD{} may be Pareto-inefficient even with ordinal preferences.

An alternative solution called probabilistic serial
(PS) was proposed by \cite{BM01} which is both envy-free and Pareto-efficient with
respect to ordinal preferences. The PS mechanism is,
however, not Pareto-efficient with cardinal preferences. This is possible, as ordinal preferences are not always sufficient for ranking the randomized (interim) allocations, (i.e., ranking of distributions over outcomes).

The mechanism we study in this paper, based on
market equilibrium from equal income, has been proposed in this context
by \cite{HZ79}, is both envy-free and Pareto-efficient even
with respect to cardinal preferences. Note that neither PS nor the market equilibrium mechanism is strategy-proof. However,
\cite{Z90} shows that for $\NumAgents \ge 3$ agents there is no mechanism that
is strategyproof, Pareto-efficient, and envy-free.

\cite{HZ79} proves that equilibrium is guaranteed to exist (see also Appendix \ref{sec:existence}), and propose an
exponential time algorithm for finding approximate
equilibrium, whereas the current paper proposes an exact algorithm for
computing equilibrium which runs in polynomial time when either the number of
agents or the number of items is constant. Most of the recent work on the
problem of assignment without money has been focused on analyzing the efficiency
of RSD and PS mechanism under cardinal and/or ordinal preferences, e.g.,
\cite{BCK11a,ASZ14}.

The main techniques used in the current paper are based on the cell
decomposition result of \cite{CellDecomp:98} which has also been used by
\cite{Devanur:2008} to derive a polynomial time algorithm for a related
market equilibrium computation problem. \etedit{We show how to find equilibria of matching problems in polynomial time when the number of agents is fixed. In Appendix \ref{Sec:FixedGoods} we also give an algorithm to find equilibria in matching markets with a fixed number of goods. } While the algorithm of \cite{Devanur:2008} \etedit{can be used for this latter case}, our algorithm is simpler: we avoid some complications (for instance their primal dual technique for checking the market clearing conditions), and we use a simpler cell decomposition theorem. The case of fixed number of agents has been studied by \cite{echenique2012}. However, they assume that the agents' utility functions are strictly concave and strictly monotone, which does not apply to our problem. They also approximate the Walrasian equilibrium, while our main goal is to find the exact value of equilibrium prices and allocations.

\section{Preliminaries}\label{sec:prelim}
In this section we review the matching problem with \etedit{additive} 
preferences and the market equilibrium solution we aim to compute.
Then we'll discuss our main technical tool, the cell decomposition technique of Basu, Pollack, and Roy \cite{CellDecomp:98}.

\paragraph{The Matching Problem}
The problem is defined by a set of $\NumItems$ items, and $\ItemCount_j\ge 0$ amount available of item $j$, and a set of $\NumAgents$ agents. The matching problem requires that we allocate exactly 1 unit of these items to all agents. The amount $\ItemCount_j$ available of each item $j$ may be very small, so the 1 unit allocated to an agent may need to be combined from small fractions of many different items. An allocation $\{\Alloc_{ij}\}$ for all agents $i$ is a feasible solution of the matching problem, if
\begin{eqnarray*}
\sum_j \Alloc_{ij}&=& 1  \mbox{   for all $i \in [\NumAgents]$}\\
\sum_i \Alloc_{ij}&\le& \ItemCount_j  \mbox{   for all $j \in [\NumItems]$}\\
\Alloc_{ij} & \ge & 0 \mbox{   for all $i \in [\NumAgents]$ and $j \in [\NumItems]$}\\
\end{eqnarray*}
We assume agent $i$ has value $\Val_{ij}$ for a unit of item $j$. So her value for a set of $\{\Alloc_{ij}\}$ amounts of each item $j$ is $\sum_j \Val_{ij}\Alloc_{ij}$, assuming $\sum_j \Alloc_{ij} = 1$.

More generally, we can require to allocate different amounts for different agents, and allow the matching constraint to be only an upper bound, that is, allocate at most 1 unit to each agent. For simplicity of presentation, in this paper we will use equal amount required for the agents, and normalize that value to 1. Further, we will assume, also for simplicity of the presentation only, that $\sum_j C_j=n$, so a feasible solution to the matching problem will fully allocate all items.

\paragraph{Fair Allocation: Matching Market.}
We use the Fisher market proposed by \cite{HZ79} to make the allocation fair. Fisher market is defined by giving each agent a unit of (artificial) money. A market equilibrium is defined by a set of prices $\Price_{j} \in \PosReals$ for each item $j$. Given a set of prices, the agent $i$'s  optimization problem can be written as follows:

\begin{align*}
    &\text{maximize}   & \sum_j & \Val_{ij} \Alloc_{ij} \\
                    &\text{subject to}   & \sum_j \Alloc_{ij} & = 1 \tag{The matching constraint} \\
    &    & \sum_j \Alloc_{ij} \Price_j & \le 1 \tag{The budget constraint} \\
                    &   & \Alloc_{ij} & \ge 0 & &\forall j \in \RangeN{\NumItems}. \notag
\end{align*}

A \emph{market equilibrium} for this market is a set of prices $\Prices$, and a feasible allocation $\Allocss$ such that (i) $\Allocs_i$ is an
optimal solution to agent $i$'s optimization problem with respect to prices $\Prices$ and her budget of 1 unit of money. Note that the requirement of market clearing that all items are allocated, is automatically satisfied due to the matching constraint and our assumption that $\sum_j \ItemCount_j=\NumAgents$.

The matching constraint in the agents' preferences creates an odd demand structure. For some prices the agent's optimization problem is not feasible, and even on prices when all agent optimization problems are feasible, the preferences do not satisfy the gross substitute property, i.e increasing price of an item may increase the demand of that item, as explained in the introduction, and an example is shown in Appendix \ref{ap:examples}.

\etedit{A natural idea to convert the problem to one with a simpler structure is to allow agents to have free disposal, i.e., if assigned more than 1 unit of item, they value the best unit. Unfortunately, this change in the model significantly changes the structure of the problem, and can result in market equilibria that are simply not feasible for the original problem.}
Consider the following market with 2 agents and 2 items.

\begin{center}
\begin{tikzpicture}[marketstyle]
    \def\marketsize{2};
    \foreach \i in {1,...,\marketsize}
        \node["\AgentID{\i}"below,circle,fill] (agent_\i) at (\i,0) {};
    \foreach \j in {1,...,\marketsize}
        \node["\ItemID{\j}"above,circle,fill] (item_\j) at (\j,1) {};
    \draw(agent_1)
        edge["$2$"] (item_1)
        edge["$1$"] (item_2);
    \draw(agent_2)
        edge["$0$"] (item_1)
        edge["$1$"] (item_2);
\end{tikzpicture}
\end{center}
Let $v_1=(2,1)$ and $v_2=(0,1)$. Assume that the prices are $\Prices=(0.5,1.5)$. When agent's \etedit{have} free disposal, these prices are equilibrium, since $\Allocs_1=(1,\frac{1}{3})$ and $\Alloc_2=(0,\frac{2}{3})$ is optimal for both agents, and also clears the market. However, with matching market preferences these prices are not equilibrium, since $I_1$ will not get any of $A_2$ and $I_2$ cannot afford all of $A_2$ so the prices are not market clearing.

We note that it is possible to express this matching market problem as a classical Fisher market with agent preferences that are piece-wise linear and concave (general PLC) functions, though non-monotone and non-separable. To do this, we first relax the matching constraint in the agents' optimization problem to requiring only that $\sum_j \Alloc_{ij} \le 1$. We will show that the market clearing condition of a Fisher market requiring that all items are allocated (or have 0 price) will help to enforce that all agents get exactly 1 unit. We can also express the agents' utility as a piece-wise linear and concave function. To do this, let $\Val^*_i = \max_{j} \Val_{ij}+\epsilon$ where $\epsilon>0$. Let agent $i$'s utility for an allocation $\Allocs_i$
$$
\min(\sum_j \Val_{ij}\Alloc_{ij}, \sum_j \Val_{ij}\Alloc_{ij} + \Val^*_i(1-\sum_j\Alloc_{ij}).
$$
When $\sum_j \Alloc_{ij}\le 1$, the first term is smaller, but when $\sum_j \Alloc_{ij}\geq 1$ the minimum is taken by the second term, so the total value strictly decreases as the allocated amount exceeds 1. Since the agent's utility decreases
 with more than one unit allocated, an optimal solution to the agent's optimization problem will allocate at most 1 unit of item to each agent, and hence if the market clears, this is also a solution to the matching market problem.

\cite{Devanur:2008} provides a polynomial time algorithm to find the equilibrium prices and allocation for the case with fixed number of goods. In Appendix \ref{Sec:FixedGoods} we give a simpler exact algorithm for this case, taking advantage of the matching structure. Our main result is to extend this to the case of fixed number of agents, instead of fixed number of items. \cite{Devanur:2008} also offers an algorithm for finding market clearing prices for with fixed number of agents, but only with \emph{separable} PLC utilities, i.e., when the utility of each agent is a separable function of the items allocated to the agents. Note that the utility function of a matching problem is necessarily not separable, as it needs to express the upper bound on the total allocation.

\paragraph{The Cell Decomposition Technique}
Our algorithms will search the space of optimal utility values of each agent, and for each possible utility value, will search through the possible structures of allocations. There are only a fixed number of agents, however, the space of possible optimal values is huge. We use a cell decomposition technique to make the search space discrete\etedit{, facilitated by a characterization of equilibria}.
The beginning of Section \ref{Sec:FixedAgents} has a more detailed outline.

A main technical tool for our work will be the following theorem concerning the way polynomials divide the space in a $d$ dimensional space. Given a set of polynomials on $M$ variables, the sign of the polynomials define an equivalence between vectors in the $M$ dimensional space $\Reals^M$, where two vectors $y$ and $y'$ are equivalent if all polynomials have the same sign on $y$ and $y'$. We call the equivalence sets of this relations the \emph{cells} of the way the polynomials divide up the space. In principle $N$ functions can divide up the space into at many as $3^N$ cells (as each polynomial can be 0 positive or negative). However, Basu, Pollack, and Roy \cite{CellDecomp:98} showed that bounded degree polynomials in small dimensional spaces define much fewer cells.

\begin{theorem}\cite{CellDecomp:98} \label{THM:CellDecomp}
If we have a set of $M$ number of variables, and $N$ number of polynomials whose degree is at most $d$, then the number of non-empty cells, and the time required to enumerate them is $O(N^{M+1})d^{O(M)}$.
\end{theorem}


\etedit{We will use this decomposition to find the equilibrium for our matching problem. To illustrate the idea, and let $\Reals^n$ be the space of all possible agent utilities. Roughly speaking the idea is as follows. If we could describe whether a set of utilities $\Util \in \Reals^n$ arises from an equilibrium by the signs of a few bounded degree polynomials in these variables, then we could use Theorem \ref{THM:CellDecomp} to enumerate all cells defined by these polynomials, and test which of the cells satisfies the condition required for being an equilibrium. Unfortunately, the equilibrium condition cannot be described this way, so we will need to introduce additional variables (helping us infer the prices and assignment, despite the fact that these are not in fixed dimensional space) to be able to carry out this plan.}

\section{Computing Market Equilibrium with Fixed Number of Agents}\label{Sec:FixedAgents}
In this section, we give an exact algorithm to find an equilibrium in the case where the number of agents $n$ is constant, and the number of different goods is an arbitrary number $m\in \mathbb{N}$, under the mild technical assumption that each agent has a unique most preferred item. More formally, for every agent $i\in [n]$, there is exactly one item $j$ such that $v_{ij}=\max_{k\in[m]}(v_{ik})$. The goal of this section is proving the following theorem.

\begin{theorem}\label{THM:FixedAgents}
Exact equilibrium (prices and allocations) in a matching market with fixed number of agents, in which agents have additive values one unit of money, and a unique most preferred item, can be found in polynomial time.
\end{theorem}

\textbf{General Outline and Techniques. }
Our algorithm searches the space of agents' optimal utilities and item prices to find an equilibrium. We divide this space into a polynomial number of cells, where each cell contains utility and price vectors that have the same properties. We use Theorem \ref{THM:CellDecomp} as the basis of the cell decomposition. The space of possible agent utilities is finite dimensional. However, since the number of items is not constant, we cannot use a separate variable for each item price. In section \ref{Sec:AgentsPriceBundleChar}, we provide a bundling technique and a characterization of the equilibrium structure that allow us to define equilibria using only a finite set of variables.

\textbf{Cell. } Consider the vector of player utilities $\Utils=\{u_i\}$, a constant dimensional space for fixed number of agents. Now consider the linear functions $u_i-v_{ij}$ for each agent $i$ and item $j$. A cell of the space of utilities $\Utils$ defined by these functions is the region of this space in which each of these functions has a fixed sign. Within each region, the items are divided for each agent $i$ into those with value above $u_i$, same as $u_i$, and below $u_i$. This division also has implications on prices: if the utilities are part of an equilibrium, the price $\Price_j$ of any item $j$ with value $v_{ij}>u_i$ will have to be above $1$. We will add further variables and polynomials, until each cell provides enough information for checking all the equilibrium properties.

\textbf{Layered Cell Decomposition. } Next we would like to add the item prices as variables. However, to keep the running time polynomial when using Theorem \ref{THM:CellDecomp}, we can only have a constant number of variables, so we cannot afford a variable for each item price. To get around this problem we will try a polynomial number of different structures for the price vector, where for any fixed structure, we can define all item prices via a fixed number of variables. To do this, for each agent we will fix a special item that is at least partially allocated to the agent.  Lemma \ref{Lem:Pricefromfew} will show that given prices for the fixed number of special items, we can infer prices for all items.

Finally, we also need to be able to find the assignment variables. We will show in Lemma \ref{Lem:ReBundling} that each pair of agents only shares a few items (at most 5), and given the set of shared items, as well as the utilities and item prices, the  allocation can be fully determined. Our algorithms iterates over all structures of specially assigned items and shared sets of items. For each of these structures, the algorithm iterates over all cells of the cell-decomposition given by agent utilities and prices of special items and the constraint (polynomials) described in the next subsections that ensure that these describe an equilibrium, and finds the ones which correspond to equilibria.\footnote{Recall that by \cite{HZ79} (see also Appendix \ref{sec:existence}) an equilibrium must exists.}

\textbf{Bundles. } Rather than thinking about individual items in isolation, it is useful to think of items in pairs. In equilibria each agent spends exactly one unit of money and gets exactly a total of one unit of items. This means that in equilibrium, if an agent gets some amount an item whose unit price is less than one, she should also get some amount of an item whose unit price is more than one. As suggested by this fact, we pair items of price below and above 1. We define a \emph{bundle} as either a single item of price 1, or fractions of two items of a total of one unit of items, where the total unit price of the bundle is exactly 1. First, we show that in a market equilibrium, the allocations of items to agents, can be rewritten as the allocation of bundles to agents.

\begin{lemma}\label{Lem:Bundling}
In an equilibrium pricing $\Prices$ with equilibrium allocation $\Allocss$, there exists a bundling $B$ of items such that
\begin{enumerate}[a.]
\setlength{\itemsep}{0pt}\setlength{\parsep}{0pt}\setlength{\parskip}{0pt}
\item Each bundle $b$ consists of at most two items. Two items $j$ and $k$ are in a bundle $b=(j,k)$ if and only if $p_j<1<p_k$. One item $j$ forms a bundle $b=(j)$ if and only if $p_j=1$.
\item Each bundle $b=(j,k)$, associated with a unique mix $0< \alpha_b < 1$ such that $\alpha_b p_j + (1-\alpha_b) p_k =1$ (recall that $p_j<1<p_k$). For bundles with one item we use $\alpha_b=1$.
\item The optimum allocation of items to agent $i$ satisfying the matching constraint can be rewritten as allocation of bundles.
\end{enumerate}
\end{lemma}
\begin{proof}
Let $\Prices$ be an equilibrium pricing and $\Allocss$ be an arbitrary allocation associated with $\Prices$. For each agent $i$, we know that the total allocation of items to $i$ and the the total cost of $i$'s items is 1 (due to the market clearing condition), i.e.
$$\sum_{j\in[m]}\Alloc_{ij}=1$$
and
$$\sum_{j\in[m]}p_{j}\Alloc_{ij}=1$$
We start rewriting $i$'s allocation with bundles. Let $y_{ib}$ be the amount of bundle $b$ that $i$ uses in the new allocation ($y_{ib}$ is zero at the beginning). At each step we consider the following cases
\begin{enumerate}
\item For every $j$ that $\Alloc_{ij}>0$, we have $\Price_j=1$: in this case all such items should be in a bundle $b=(j)$ ($y_{ib}=x_{ij}$), so our claim is correct.
\item There exists $j$ that $\Alloc_{ij}>0$ and $\Price_{j}>1$: This means that there exists another item $j'$ such that $\Price_{j'}<1$ and $\Alloc_{ij}>0$, otherwise if $i$ gets $t$ unit of items, her total cost cannot be $t$ (note that $t=1$ here). So by the definition, $j$ and $j'$ should be in a bundle $b=(j,j')$.
\end{enumerate}

Let $z=\min(\frac{x_{ij}}{\alpha_b},\frac{\Alloc_{ij'}}{1-\alpha_b})$. In the second case, we increase $y_{ib}$ by $z$, and reduce $\Alloc_{ij}$ by $z\alpha_b$ and $\Alloc_{ij'}$ by $z(1-\alpha_b)$. This means that both the total cost of remaining allocation of $\Allocss$ and the total remaining items in allocation of $\Allocss$ decrease by $z$. Note that by doing this the total allocation of $i$ (counting her allocation in $\Allocs_i$ and $y_i$) does not change. Furthermore, either $\Alloc_{ij}$ or $\Alloc_{ij'}$ becomes zero.

If we repeat this process, $y_i$ gives us a way to rewrite the allocation of items to $i$ as an allocation of some bundles ($B_i$) to $i$. If we repeat this for all the agents, we get what we want.
\end{proof}

For each agent $i$, the value of a bundle $b=(j,k)$ is $v_{ib}=\alpha_b v_{ij} + (1-\alpha_b) v_{ik}$, while the value of a single item bundle $b=(j)$ is just the value $v_{ib}=v_{ij}$. Let $B_i\subseteq B$ the set of bundles of maximum value for agent $i$ (called $i$'s optimum bundles).

\begin{corollary}\label{Lem:Bundlingd}
The optimum allocation of items to agent $i$ is any allocation of bundles using only bundles in $B_i$. Furthermore, utility of $i$ in the equilibrium is $u_i=v_{ib}$ for the bundles $b\in B_i$.
\end{corollary}
\begin{proof}
We prove this by contradiction. Assume that $\mathbf y$ is the allocation of bundles to agents. Also assume that there exists an agent $i$ and a bundle $b$ for which $y_{ib}>0$ but there exists another bundle $b'$ such that $v_{ib}<v_{ib'}$. Since the unit price of the bundles is $1$, and there is one unit of items in them, $i$ can trade her share of $b$ for the same amount of $b'$ and increase her value. This is a contradiction.

Now since the unit values of bundles that $i$ uses are the same, and $i$ gets exactly one unit of bundles, her utility in equilibrium is equal to the unit value of these bundles. So the second claim is also true.
\end{proof}

A key observation for using the bundles to define prices in Lemma \ref{Lem:Pricefromfew} is the following exchange property of optimal bundles.

\begin{lemma}\label{Lem:Bundlinge}
If $b=(j,k)$ and $b'=(j',k')$ are optimum bundles of agent $i$ (are in $B_i$), such that $p_j,p_{j'}<1<p_k,p_{k'}$, then $b''=(j,k')$ and $b'''=(j',k)$ are also in $B_i$.
\end{lemma}
\begin{proof}
We prove this lemma by using contradiction. W.l.o.g assume that $b_3 \notin B_i$. If in equilibrium $i$ trade her bundles to get $y_{ib_1}>0$ unit of $b_1$ and $y_{ib_2}>0$ unit of $b_2$, then by lemma \ref{Lem:Bundlingd}, her value will not change. Now, since $\Alloc_{ij}>0$ and $\Alloc_{ik'}>0$, similar to proof of lemma \ref{Lem:Bundling}, we can rewrite her allocation so that it includes some of $b_3$, and similarly rewrite the remaining allocation of $i$ with other bundles. Since the unit value of $i$ for $b_3$ is less than value of $i$ in equilibrium, this is a contradiction with corollary \ref{Lem:Bundlingd} and we are done.

\end{proof}



\subsection{Characterizing the Prices and Optimum Bundles with Polynomials}\label{Sec:AgentsPriceBundleChar}
In this subsection we define a set of variables and polynomials that \etedit{help us determine agent utilities and prices of all items at equilibrium. We consider assignments in the next subsection}. 

For each agent $i$, we define a variable $u_i$ which is $i$'s utility in the equilibrium.
By Lemma \ref{Lem:Bundling}, we know that any equilibrium allocates bundles to agents, where each agent only gets one unit of her optimum bundles $B_i$. 
Since we did not define variables for the 
prices yet, we cannot use prices to define bundles, so we start by defining a set of item bundles for each agent just based on the fact from Corollary \ref{Lem:Bundlingd} that optimum bundles must give value $u_i$.

\textbf{Candidate Bundles.} For each agent $i$, we define a \emph{candidate bundle} to be the items whose value is $u_i$, or the pair of items $j$ and $k$ such that $v_{ij}<u_i<v_{ik}$, so there exists a unique $0<\alpha^i_{jk}<1$ such that $\alpha^i_{jk}v_{ij} + (1-\alpha^i_{jk})v_{ik}=u_i$. Note that the optimum bundles of $i$ also satisfy this constraint (by Corollary \ref{Lem:Bundlingd}). This means that the optimum bundles of an agent is a subset of her candidate bundles. In addition, the price of optimum bundles is exactly 1.

In order to find the set of candidate bundles of agent $i$, we  define a polynomial $v_{ij}-u_i$ for each agent $i$ and item $j$. This way, each cell tells us for each item $j$, whether $v_{ij}<u_i$, $v_{ij}=u_i$ or $v_{ij}>u_i$.
For any two items $j,k\in[m]$, $j$ and $k$ form a candidate bundle if $v_{ij}<u_i$ and $v_{ik}>u_i$. Similarly if for an item $j$, if $u_i=v_{ij}$ then $j$ alone forms a candidate bundle. By the information provided by each cell, for each agent $i\in[n]$ and item $j,k\in [m]$ that form a candidate bundle, we define the ratio for the candidate bundle to be $\alpha^i_{jk}=\frac{u_i-u_k}{u_j-u_k}$.

Not all the candidate bundles of agent $i$ are in the set of her optimum bundles, since the price of optimum bundles should be exactly 1. We first observe that the unit price of a candidate bundle cannot be less than 1. This property of candidate bundles is useful in proof of lemma \ref{Lem:Pricefromfew}, allowing us to infer all prices from prices on only a few items.

\begin{lemma}\label{Lem:CanBunPrice}
In an equilibrium with prices $\Prices$, 
\etedit{all candidate bundles have price at least $1$}.
\end{lemma}
\begin{proof}
We prove this by contradiction. Assume that there exists an agent $i$ such that the price of one her candidate bundles $b$ is less than one. Assume that $i$'s most preferred item is $j$. We have the following cases
\begin{itemize}
\item $b=(z)\neq (j)$: In this case, since $i$'s most preferred item is unique, there exists an $\epsilon>0$ such that if $i$ gets $1-\epsilon$ unit of item $z$ and $\epsilon$ unit of item $j$, her total price is still less than 1, but her utility is more than $v_{iz}$. This is a contradiction since by the definition of candidate bundles, $i$'s utility in equilibrium cannot be more than $v_{iz}$.
\item $b=(j)$: In this case, we claim that $\Allocss$ does not allocate any item to $i$ other than $j$. Assume this is not true. If $i$ trades whatever she gets in the equilibrium with $j$, her price will be less than 1, and since the maximum value for $i$ is unique, her utility will increase. So it should be the case that $i$ does not spend all her money and market will not be cleared. This is a contradiction with the assumption that the prices are market clearing.
\item $b=(k,z)$: Since the maximum value item of $i$ is unique, her value $b$ should be less than $v_{ij}$. Now since the unit price of $b$ is less than one, there exists $\epsilon>0$ such that $\epsilon p_j + (1-\epsilon) \Price_b <1$. However, if $i$ gets $\epsilon$ from item $j$ and $1-\epsilon$ from bundle $b$, his utility is more than $v_{ib}$. This is a contradiction with the definition of candidate bundles.
\end{itemize}

So all the possible cases reach a contradiction and we are done.
\end{proof}


Next we wish to find the prices for all items. We will show that if one knows for each agent her utility, the price of only one item in her optimum bundles, and we use the set of candidate bundles defined above, we can find a the price of all the items which are in one of her optimum bundles.

\begin{lemma}\label{Lem:Pricefromfew}
Consider an equilibrium where we know for each agent $i$, the utility $u_i$ of the agent, and the price of a single item $j$ which is in a bundle of $B_i$ with two items. Using these values, and the notion of candidate bundles defined above, we can find the price of all items in polynomial time.
\end{lemma}
\begin{proof}
The key for finding the prices is the observation that if for a bundle $b=(j,k)$, we have $\alpha_b$ and $\Price_j$, then we can find $\Price_k$. This fact, combined with lemma \ref{Lem:Bundlinge}, imply that if for each agent $i$, we know the price of one item in one of her optimum bundles, then we can find the price of all the items in her optimum bundles. The only problem is that we do not know which one of her candidate bundles is also one of her optimum bundles.

Assume that for each agent $i$, we know a good $g_i \in [m]$ is in one of her optimum bundles and we have a variable for $g_i$'s price. To find a formula for price of other items with the variables we have defined so far, consider the following game. Assume that $i$ is playing a game, in which she wants to maximize the number of her optimum bundles. She can participate in this game by proposing a price for all the items in her candidate bundles, knowing $\Price_{g_i}$ and her candidate bundles. W.l.o.g assume that $\Price_{g_i}<0$. For each of her candidate bundles $b=(g_i,j)$, $i$ first proposes price $\Price_j=\frac{1-\alpha_b{g_i}}{1-\alpha_b}$ for item $j$, then for each of her candidate bundles $b'=(k,j)$, she proposes price $\Price_k=\frac{1-(1-\alpha_{b'})\Price_j}{\alpha_{b'}}$ for item $k$. Finally she proposes price $1$ for all the items which form a bundle alone. By doing this, if any of her proposed prices is chosen, that item will be in her optimum bundles. Now, for each item $j$, we, the game coordinator, choose the maximum price for $j$ among all the prices which were proposed by the agents for $j$ and set that to be the price of item $j$.

Note that in equilibrium, we have to choose the maximum proposed price, since if we choose less than that, the agents with higher proposals will have candidate bundles whose price is less than 1, this is a contradiction with lemma \ref{Lem:CanBunPrice}.
\end{proof}

In order to use the above lemma, we should be able to do two things: (i) For each agent $i$, \etedit{select} an item $j$ that is in \etedit{one of} $i$'s optimum bundle \etedit{with two items} \etedit{(if such an item exists)}, and \etedit{set} its price $\Price_j$. (ii) For each item $j$, find the maximum proposed price among all the proposed prices for that item\etedit{, with proposed price of 1 of items in single item candidate bundles for any agent (including agents with no special item)}. We can do the first task by checking all possible assignments of the special items with defining $O(m^n)$ separate equilibrium structures for each possible selection of one special item assigned to each agent, and checking them separately. Since the number of agents is constant, the number of different equilibrium structures is polynomial. For each case, we use $O(n)$ variables, one for agent utility, and one for the price of \etedit{the proposed special items for agents}. To define prices of other items, we use candidate bundles for each agent to define candidate prices, add polynomial comparing the expressions for candidate prices. The actual price of the item is the highest of all prices as shown in the proof of Lemma \ref{Lem:Pricefromfew}, which is now set uniquely in each cell. Note that if an agent $i$ is proposing a price for an item $j$ which is higher than the proposed price of another agent $i'$ and $j$ is the special item of $i'$, then this cell cannot contain equilibria.


\begin{lemma}\label{Lem:SecAgPrBun}
Consider the space of at most $2n$ variables including agents' utilities and price of one item in each agent's optimum bundles. Now we add $O(mn^2)$ polynomials: comparing utilities to item values, and comparing candidate item prices using candidate budges, as defined in Lemma \ref{Lem:Pricefromfew}, the sign of these polynomials gives us a formula for the price of each item as well as the set of optimum bundles of each agent.
\end{lemma} 
\subsection{Characterizing the Equilibria}\label{Sec:AgentsEqChar}
In this section, we add a set of new variables and polynomials to the set of variables and polynomials defined in Section \ref{Sec:AgentsPriceBundleChar}, in order to determine whether each cell contains equilibria. The new variables will help us define assignments. We cannot directly define a variable representing the allocation of all goods/bundles to agents, since the number of these is not constant.

\etedit{The key} 
idea is to show that for every equilibrium pricing, there is a specific allocation of items to agents where the number of items which is being shared between multiple agents is very small, and the allocation has a special structure. This helps us to significantly reduce the number of variables needed to define allocations. Consider two agents $i<j$, and all items of price $p_k<1$ in sorted by price as shown by Figure \ref{Fig:EqChar2}. We will show that there is an equilibrium allocation with \etedit{only two of these items shared between $i$ and $j$, and} the structure indicated by the figure, and the analogous structure for items of price $p_k>1$.
If we know the shared items, this structure helps us with finding the owner of the items which only have a single owner in the allocation, hence finding the allocation of each agent using allocation variables only for shared items. Before stating the properties of this allocation in Lemma \ref{Lem:ReBundling}, we have to define some notations.

\begin{figure}
\centerline{\includegraphics[scale=0.6]{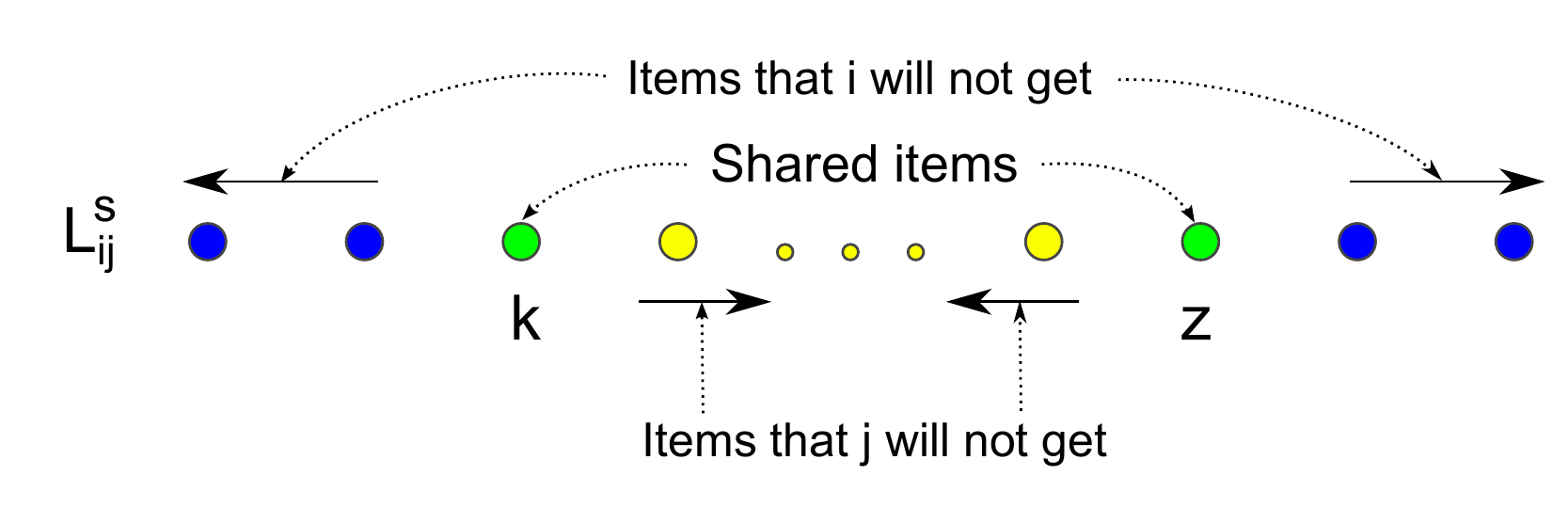}}
\caption{The nodes are the items in $S$ (one side of $G$) that are in a optimum bundle of both $i$ and $j$ ($i<j$), which are sorted by their prices. The figure shows the items that $i$ and $j$ share (green nodes), the items that of $i$ or $j$ will not get in the special allocation (blue and yellow nodes).}
\label{Fig:EqChar2}
\end{figure}

Consider a bipartite graph $G=(S,T,E)$ in which vertices in $S$ are the items with price less than $1$ and vertices in $T$ are the items with price more than $1$. For simplicity, we sort  items in each side of $G$ by their price and break ties with items' indexes. Figure \ref{Fig:EqChar2} is one side of this bipartite graph. For each bundle, we put an edge in the graph which connects the two items in the bundle.  For every pair of agents $i$ and $j$, let $B^S_{ij}=S\cap (V(B_i) \cap V(B_j))$ and $B^T_{ij}=T\cap (V(B_i) \cap V(B_j))$. Furthermore, let $H$ be the set of items whose price is 1 and let $B^H_i$ be the optimum bundles of agent $i$ which has exactly 1 item.

For each pair of agents $i$ and $j$, let $L^S_{ij}$ and $L^T_{ij}$ be the list of items in $B^S_{ij}$ and $B^T_{ij}$ sorted by their price in increasing order (break ties with index of the items), respectively, the set of items of price above 1 (and below 1) that are part of optimum bundles for both $i$ and $j$. Figure \ref{Fig:EqChar1_Main} is showing a part of the above bipartite graph for a pair of agents.

\begin{figure}
\centerline{\includegraphics[scale=0.6]{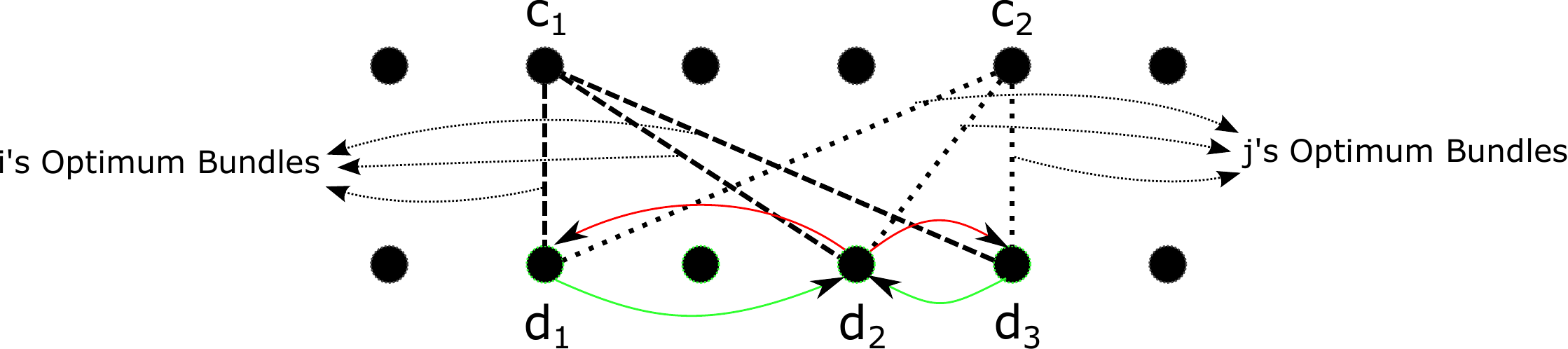}}
\caption{The arrows show the trades of items between agents $i$ and $j$ in the proof of Lemma \ref{Lem:ReBundling}.}
\label{Fig:EqChar1_Main}
\end{figure}

The following Lemma claims that there is an equilibrium  allocation of a special structure suggested by Figure \ref{Fig:EqChar2}: on each side of the bipartite graph the agents share at most two items, and the agent with lower index only gets items between the two shared items, while the agent with higher index only gets items outside this interval \etedit{as shown by Figure \ref{Fig:EqChar2}}.

The main idea of the proof is that \etedit{given} any equilibrium allocation, we can make each pair of agents that violate the properties trade their items, as illustrated by the Figure \ref{Fig:EqChar1_Main}. 
Note that the running time of this trading process is not important, since we only use it to show that for any equilibrium pricing an allocation with the desired structure exists.

\begin{lemma}\label{Lem:ReBundling}
For every equilibrium pricing $\Prices$, there exists an equilibrium allocation $\Allocss$ of items to agents such that for every pair of agents $i$ and $j$ ($1\leq i<j\leq n$)
\begin{enumerate}[1.]
\setlength{\itemsep}{0pt}\setlength{\parsep}{0pt}\setlength{\parskip}{0pt}
\item There are at most $2$ items that $X$ is allocating to both $i$ and $j$ on each side of $G$. Furthermore, if $k$ and $z$ are two items in $S$ ($T$) such that $p_k\leq p_z$ and $i$ and $j$ are sharing $k$ and $z$, $i$ \etedit{only gets items between $k$ and $z$ in the order sorted by price, while}
    $j$ does not get any of the items from $B^S_{ij}$ ($B^T_{ij}$) whose whose position in $L^S_{ij}$ ($L^T_{ij}$) is between \etedit{$k$ and $z$ in the order}
    (see Figure \ref{Fig:EqChar2}).
\item There is at most one item $k$ in $B^H_{i} \cap B^H_{j}$ that \etedit{is shared by $i$ and $j$ in $\Allocss$, 
    and} $i$ only get items from $B^H_i$ whose index is lower than $k$ and $j$ only gets items whose index is higher than $k$.
\end{enumerate}
\end{lemma}
\begin{proof}
Suppose that we have an allocation of bundles to agents $\mathbf y$ in an equilibrium. We want to reallocate these bundles so that it satisfies the conditions of the lemma.

Assume that there exist two agents $1\leq i<j \leq n$, such that there are tree bundles $a_1=(c_1,d_1)$, $a_2=(c_1,d_2)$ and $a_3=(c_1,d_3)$ such that $a_1,a_2,a_3 \in B_i$ and three bundles $b_1=(c_2,d_1)$, $b_2=(c_2,d_2)$ and $b_3=(c_2,d_3)$ such that $b_1,b_2,b_3 \in B_j$. Suppose that $p_{d_1}\leq p_{d_2} \leq p_{d_3}$ and $Y\mathbf y$ is allocating $y_{id_1},y_{id_3}>0$  of $a_1$ and $a_3$ to agent $i$, and $y_{jd_2}>0$ of $d_2$ to agent $j$. Furthermore, assume that $y_{ic_1},y_{jc_2}>0$.

Since  $\Price_{d_1}\leq \Price_{d_2} \leq \Price_{d_3}$, there is $0<\beta<1$ such that $\beta \Price_{d_1} + (1-\beta) \Price_{d_3}=\Price_{d_2}$. So if we remove $\beta z$ from $y_{id_1}$ and $(1-\beta) z$ from $y_{id_3}$ and add $z$ to $y_{id_2}$ then the total cost of agent $i$ will be the same. Similarly, if we add $\beta z$ to $y_{jd_1}$ and $(1-\beta) z$ to $y_{jd_3}$ and remove $z$ from $y_{jd_2}$, then the total cost of $j$ will not change. Furthermore, it is easy to see that doing this does not affect the matching constraints. Figure \ref{Fig:EqChar1} demonstrates this procedure.

\begin{figure}
\centerline{\includegraphics[scale=0.5]{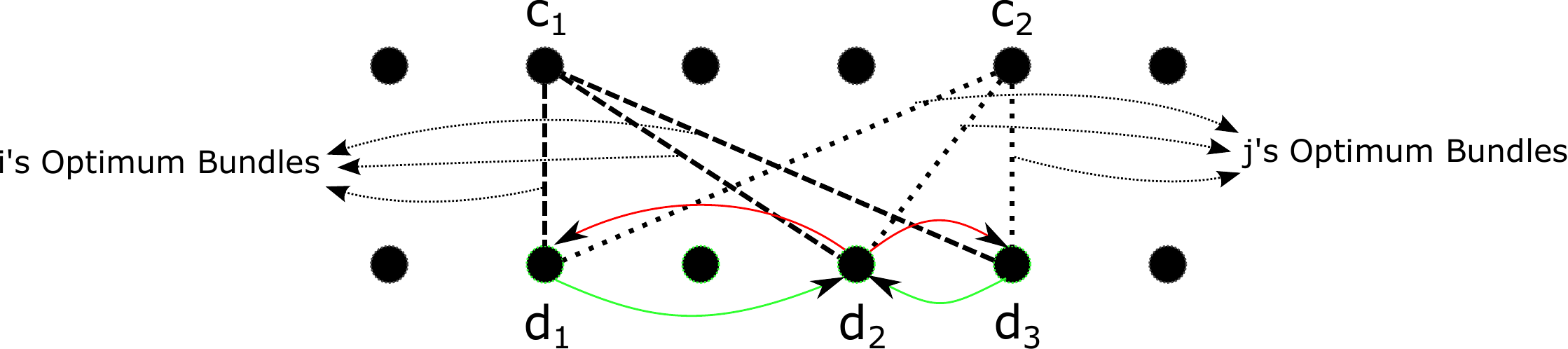}}
\caption{Red and green arrows shows the trades of items between agent $i$ and $j$ in proof of lemma \ref{Lem:ReBundling}.}
\label{Fig:EqChar1}
\end{figure}

Now, we have to show that doing this does not change the utility of $i$ and $j$. From the assumption that $Y\mathbf y$ is an equilibrium allocation, it follows that by doing this, the utility of $i$ and $j$ cannot increase. Assume that by doing this utility of agent $i$ decreases. We have
\begin{align*}
u_i&=\alpha_{a_2}v_{id_2} + (1-\alpha_{a_2})v_{ic_1}\\
&< \alpha_{a_2}(\beta v_{id_1}+ (1-\beta)v_{id_3}) + (1-\alpha_{a_2})v_{ic_1}\\
&= \beta \alpha_{a_2}v_{id_1} + (1-\beta)\alpha_{a_2}v_{id_3} + (1-\alpha_{a_2})v_{ic_1}
\end{align*}
similarly we have
$$\beta \alpha_{a_2}\Price_{d_1} + (1-\beta)\alpha_{a_2}p_{d_3} + (1-\alpha_{a_2})\Price_{c_1}=1$$
This means that if agent $i$ only use these three items with these ratios, he gets one unit of item, spends at most 1 unit of money and her utility will be more than $u_i$. This is a contradiction. The argument for agent $j$ is similar to the argument for agent $i$.

For each pair of agents $1\leq i<j \leq n$, we find the cheapest item $k$ and the most expensive item $z$ in $S$ that $i$ owns, if there is an item $r$ in between $k$ and $z$ in the ordered list of items in $S$ that $j$ owns, we switch the allocation we switch the ownership of these items until, one of the following cases happen

\begin{itemize}
\item $j$ runs out of item $r$.
\item $i$ runs out of item $k$.
\item $i$ runs out of item $z$
\end{itemize}

Let $\Phi$ be a potential function which is equal to $|S|$ minus the the position of highest positioned item that $i$ owns in the ordered list, plus the position of the lowest positioned item that $i$ owns in the ordered list, plus the number of items for which $j$ is a shareholder and are between the the two in the ordered list. If we repeat the above procedure, each time this potential function will decrease by 1. Since the potential function is always non-negative, we cannot continue the above procedure forever. This means that at after some iterations, we will reach an equilibrium allocation which satisfies the first condition for this pair of agents.

We repeat this procedure until such pair of agents and set of items with these properties do not exist. We also do the same to the allocation of items on the other side of $G$ ($T$) to agents. For the items in $B^H_{i} \cap B^H_{j}$, one can also transfer items between $i$ and $j$ to satisfy the second condition by finding two items that violate the condition and switch their ownership. The procedure defined for satisfying these conditions are an easier version of the previous procedure.
\end{proof}

From now on, we focus on finding and characterizing the specific equilibrium allocation which is guaranteed to exist by this lemma. By using Lemma \ref{Lem:ReBundling}, for every pair of agents $i,j\in[n]$ we can use 5 item indexes $f^S_{ij}$, $r^S_{ij}$, $f^T_{ij}$, $r^T_{ij}$, and $h_{ij}$ that tell us which items are shared by agent $i$ and $j$.
We can assume that we know what are these shared items by simply checking all the possible $O(m^{5n^2})$ of these combinations. 
At the start of the algorithm, we fix the 5 shared items for each pair of agents and an item in optimum bundle of each agent (which we discussed in Section \ref{Sec:AgentsPriceBundleChar}). We call this set of items associated with each agent and pairs of agents the \textit{structure} of the equilibrium. For each such a structure, we will aim to decide if there is an equilibrium with the given structure.

When considering equilibria of a given structure, we define variables for the allocation of the at most $5n^2$ shared items, but do not define variable for allocations of other items. Next we show that given the allocations of the shared items, we can (i) infer allocations of all other items using the structure of Lemma \ref{Lem:ReBundling}, and (ii) can also define polynomials whose sign will tell us if there is an equilibrium with the given structure and allocation of shared items.

Since all the items should be sold in the equilibrium, for each agent $i$, all the items that are in only in $i$'s optimum bundles, should get allocated to $i$. Second, we have defined a variable that indicates the share of each agent for the shared items. The only thing left is to consider items that are in the set of the optimum bundles of multiple agents, but are not shared by these agents.

Lemma \ref{Lem:ReBundling} helps us find allocation of this set of items. We start with the agent with the lowest index (agent 1) and one side of $G$, say $S$. The items that agent 1 gets, are the ones that satisfy all the constraints given to us by second part of the lemma. In order to check whether for an item $j$ all these constraints hold, for each agent $i$, we can look at the position of item $j$ in $L^S_{1i}$. If $j$'s position is between position of $f^S_{1i}$ and $r^S_{1i}$ in the list, for all the choices of $i>1$, then it satisfies all the constraints. By lemma \ref{Lem:ReBundling}, we know this is a necessary and sufficient condition for $j$ to get allocated to agent 1. We do this for $S$ and $T$ separately, and remove the items that get allocated to agent 1. Now, we repeat this procedure for agent 2, but only check the constraints for agents $i>2$, then remove the items that agent 2 gets. We repeat this for rest of agents based on their indexes. We can do the same procedure for items in $H$ to find which agent is getting what item. Now, we are ready to exactly specify what are the necessary and sufficient conditions for the prices in each cell to be the equilibrium prices. This process is summarized in the following lemma.

\begin{lemma}\label{Lem:AEqCond}
 Consider a structure of special items for agents, and shared items between agents (as defined after Lemma \ref{Lem:ReBundling}). Now consider a cell in the space of variables of agent utilities, prices of the special items, and allocation shares of the shared items, defined by the polynomials that help define prices of all items. The prices and allocation of this cell forms an equilibrium, if and only if the allocation defined above satisfies the following constraints
\begin{enumerate}[1.]
\setlength{\itemsep}{0pt}\setlength{\parsep}{0pt}\setlength{\parskip}{0pt}
\item All the items get fully allocated to agents.
\item Each agent gets exactly one unit of items.
\item For each agent $i$, the total cost of buying the items allocated to $i$ is exactly 1.
\item For each agent $i$, each of the items allocated to $i$ is in one of her optimum bundles.
\end{enumerate}
\end{lemma}
\begin{proof}
$\Rightarrow$ The first, second and third condition follow from the market clearing conditions. The fourth condition directly follows from part c of lemma \ref{Lem:Bundling}.

$\Leftarrow$ From the first and second condition, we know that the pricing and allocation are market clearing. From the second, third and fourth condition, and the argument in proof of part c of lemma \ref{Lem:Bundling} we know that the allocation of items to each agent can be rewritten as her optimum bundles to her. So, from the definition of optimum bundles we know that the allocation of items to agents is optimal. Therefore, the prices and allocation are in equilibrium state.
\end{proof}

The final thing we need to do is to define a set of polynomials for checking the above conditions. \etedit{The first conditions holds for an item $j$, if summing up the share of each agent for that item, the sum is equal to $C_j$.}
This can be handled by adding one polynomial for each item $j$.

\etedit{The second condition holds for an agent $i$, if when
we sum all the items (including the proportion of the shared items) that $i$ gets, this sum  is exactly 1. } So we can also check this condition by adding a polynomial for each agent. Note that we can do this since we explained how to find out what is the allocation of items to agents for this cell.

\etedit{The third condition holds by Lemma  \ref{Lem:AEqCond} if multiplying the share of each agent for an item by its price and summing over all items, we get 1 for all the agents.  We} can do this by defining a polynomial for each agent and checking its sign. The fourth condition is guaranteed to hold \etedit{by definition}. 
The following theorem summarizes how we equilibria are cells of the constraints discussed throughout this section.

\begin{theorem}\label{Thm:FixedAgentsEqChar}
Consider \etedit{an equilibrium} structure, 
and the space of the $O(n)$ variables for agents' utilities and price of one special item in each agent's optimum bundle, the $O(n^2)$ variables for allocation of shared items between the agent. Divide this space into cells by the signs of the polynomials defined in the previous section, along with the $O(m^2+n)$ polynomials defined just above, for checking the existence of the special equilibrium allocation. The sign of these polynomials fully determines either that \etedit{the} vectors in this cell can be extended to form equilibria.
\end{theorem}

Using this structure Theorem allows us to prove Theorem \ref{THM:FixedAgents}.

\textbf{Proof of Theorem \ref{THM:FixedAgents}. } We use theorem \ref{THM:CellDecomp} as the base of our algorithm. \etedit{We start by
fixing the structure of the equilibrium,  selecting a single item from the optimum bundles of each agent, and selecting fixing 5  items shared} for every pair of agents. 
We check all the possible combinations, at most  $O(m^{n+5n^2})$ options, which is polynomial in $m$ \etedit{for fixed $n$.} 

For a given equilibrium structure, we use the $O(n)$ variables, the agent utilities, and $O(mn)$ polynomials defined in section \ref{Sec:AgentsPriceBundleChar} to find a set of candidate bundles for each agent. Then we use an additional $O(n)$ variables, the prices of the special items for each agent, and $O(mn^2)$ polynomials in order to find a formula for the prices and the set of optimum bundles of each agent in each cell. Finally, we use the last set of $O(n^2)$ variables, the assignments of shared items, and $O(m^2+n)$ additional polynomials, in order to check whether the set of prices in the feasible cell are equilibrium prices with the given structure. 
The degree of the defined polynomials is polynomial. We check all the non-empty cells of the resulting system, taking time polynomial on $m$ for any fixed $n$ by Theorem \ref{THM:CellDecomp}. Since the equilibrium exists, it should be in one of the non-empty cells.

Finally, if the prices of the cell are equilibrium prices, we take any vector from that cell, and extend it to get an equilibrium pricing and allocation. After we have the equilibrium prices $\Prices$,  we can also find each agent's allocation by finding a solution of the following set of inequalities.

\begin{align*}
&\sum_{b \in B_i}\Alloc_{ib}=1 &\forall i\in[n]\\
&\sum_{i}\big(\sum_{b:b=(j,k)}\alpha_b\Alloc_{ib} + \sum_{b:b=(k,j)}(1-\alpha_b)\Alloc_{ib} + \sum_{b:b=(j)} x_{ib} \big) =C_j &\forall j\in[m]\\
&\Alloc_{ib} \geq 0 &\forall i\in[n], b\in B\\
\end{align*}

In which, $\Alloc_{bi}$ is the amount of bundle $b$ used by agent $i$. Note that since we know this bundling is associated with an equilibrium, the feasible region of the above inequalities is not empty. Finally, the allocation of each agent $i$ for item $j$ in this equilibrium is
$$\Alloc_{ij}=\sum_{b:b=(j,k)}\alpha_b\Alloc_{ib} + \sum_{b:b=(k,j)}(1-\alpha_b)\Alloc_{ib} + \sum_{b:b=(j)} \Alloc_{ib}$$



\bibliographystyle{alpha}
\bibliography{BibFile}

\appendix
\section*{APPENDIX}
\setcounter{section}{0}
\section{Fixed Number of Goods}\label{Sec:FixedGoods}
In this section, we give a polynomial time algorithm which finds the exact value of equilibrium prices and allocations when the number of goods is constant $m$ and the number of agents is an arbitrary number $n\in \mathbb{N}$. In order to find an equilibrium, our algorithm uses a cell decomposition technique which uses the algorithm proposed in \cite{CellDecomp:98}.

The goal of this section is to prove the following theorem by providing a polynomial time algorithm.

\begin{theorem}\label{THM:GoodsMain}
Finding an equilibrium of a market with fixed number of goods in which agents have additive values, one unit of money and matching constraints can be done in polynomial time.
\end{theorem}

\subsection{Characterizing the Bundles with Polynomials}\label{Sec:GoodsBundleChar}
In this section, we define a set of variables, which represent the prices, and polynomials so that the sign of polynomials in each cell determines the bundling and the optimum bundles of each agent for the set of prices in that cell.

First, for each item $j$ we define a variable $\Price_j$ to represent its price. We define $m^2$ polynomials $\Price_j - \Price_k$ for all $j,k\in[m]$, $\Price_j$ for all $j\in[m]$ and $\Price_j-1$ for all $j\in [m]$ that give us the order of items' prices, check whether price is non-negative and whether $\Price_j$ is less than, greater than or equal to $1$, respectively. Recall lemma \ref{Lem:Bundling}. Instead of finding equilibrium allocation of items to agents, we focus on finding the allocation of bundles to agents.

By lemma \ref{Lem:Bundling} we know that each bundle should have price $1$. It means that if items $j$ and $k$ are in a bundle $b$, then $\Price_b=\alpha_b \Price_j + (1-\alpha_b) \Price_k =1$. We can rewrite this equation to get $\alpha_b$ from $p_j$ and $\Price_k$, i.e. $\alpha_b=\frac{1-\Price_k}{\Price_j-\Price_k}$. So for every $j,k\in[m]$, we define $q_{jk}=\frac{1-\Price_k}{\Price_j-\Price_k}$, and when $0 \leq q_{jk} \leq 1$, it means $j$ and $k$ form a valid bundle. Note that we don't put $j$ and $k$ in a bundle unless $\Price_j \neq \Price_k$. Since in each cell we know the order of prices, we can check whether $\Price_j= \Price_k$, $\Price_j>p_k$ or $\Price_j<\Price_k$, then for the last two cases, check whether $0 \leq  1-\Price_k \leq \Price_j-\Price_k$ and $0 \geq  1-\Price_k \geq \Price_j-\Price_k$ respectively. The only thing left now is to check the bundles with only one item. Since an item forms a bundle if and only if its price is exactly 1, we can do this by checking sign of $\Price_j-1$ for all the items. It is easy to see that we can find whether each of these conditions hold by defining $O(m^2)$ polynomials and checking their sign in each feasible cell.

Now, each output cell of the cell decomposition algorithm gives us the valid bundling associated with the signs of the polynomials in that cell. We also wish to find the optimum bundles for each agent. Note that for a valid bundle $b=(j,k)$, $q_{jk}=\alpha_b$. By corollary \ref{Lem:Bundlingd}, we know that value of each agent should be maximum for all of her optimum bundles. So given that for each bundle $b$, we can write $\alpha_b$ in terms of the prices, we can also write the value $v_{ib}$ of each agent $i$ for each bundle $b$. So to find out which of the valid bundles is optimum, we can check the sign of $v_{ib}-v_{ib'}$ for all $i\in[n]$ and $b,b' \in B$, and order the value of agents for bundles. This can be done by defining $O(nm^4)$ number of polynomials. By adding these polynomials, each cell also tells us what is the order of each agent over all the bundles, hence we can find the optimal bundles of each agent.

We can sum up the above arguments in the following lemma.
\begin{lemma}\label{Lem:SecGBun}
For each cell, defined by the above $O(nm^4)$ polynomials and $O(m)$ variables, we can use the signs of the defined polynomials in order to find the bundles and the optimum bundles of each agent for the set of prices in that cell.
\end{lemma}

\subsection{Characterizing the Equilibria}\label{Sec:GoodsEqChar}

In the previous section we described how to define polynomial number of polynomials with constant number of variables (prices) such that each cell specifies all the valid bundles and the optimum bundles of each agent. We also showed how we can use these variables to get formulate $\alpha_b$ ($q_{jk}$) for each bundle $b=(j,k)$. In this section we describe how we can define a set of polynomials whose sign tells us if a given bundling satisfies the equilibrium conditions.

One difficulty is that if we want our algorithm to run in polynomial time, we cannot directly define a variable $\Alloc_{ij}$ for each agent $i$ and each item $j$, since the number of variables needs to remain constant. In order to get around this problem, we define a variable for each bundle instead that shows how much this bundle is used in equilibrium, and use these variables to check equilibrium conditions.

Assume that we have an equilibrium $E(\Allocss,\Prices)$ and its valid bundling $B$. Let $\Allocs_b$ be the total amount of bundle $b \in B$, which is used in this equilibrium. For a subset of agents $S$, let $B_S$ be the union of all the optimum bundles of agents in $S$, i.e. $B_{S}=\bigcup_{i\in S}B_i$. In order to characterize equilibrium bundles with variables $x_b$ and polynomials, we use the following lemma. 

\begin{lemma}\label{Lem:GEChar}
A pricing and its corresponding bundling is an equilibrium if and only if there exists $\Allocss\in \PosReals^{|B|}$ such the following hold
\begin{itemize}
\item[a.] For every subset of agents $S$, we have $\sum_{b \in B_S} \Allocs_b \geq |S|$.
\item[b.] For every item $j$, $\sum_{b:b=(j,k)} \alpha_b\Allocs_b + \sum_{b:b=(k,j)} (1-\alpha_b)\Allocs_b + \sum_{b:b=(j)}\Allocs_b = C_j$.
\end{itemize}
\end{lemma}
\begin{proof}
\pjcomment{Added the missing proof }The proof of this lemma is similar to proof of the Hall's theorem. On one side of the bipartite graph we have agents with capacity 1 and on the other side we have the bundles with capacity $x_b$ (for each bundle $b\in B$). Furthermore, there is an edge between bundle $b$ and agent $i$ iff $b \in B_i$.

$\Rightarrow$ This case directly follows from lemma \ref{Lem:Bundling} and the market clearing conditions.

$\Leftarrow$ Assume that there is an $X\in R^{|B|}$ such that a and b hold. We have to show that there exists an allocation of bundles to agents such that every agent gets exactly 1 unit of her optimum bundles. We use contradiction. Assume that this allocation does not exists. Let $Y$ be an allocation which allocates maximum total amount of bundles to agents among all the valid allocations. Let $i$ be an agent for which $Y$ allocates less than 1 unit of her optimum bundles. If there is an augmenting path from $i$ to a bundle $b$ which is not fully consumed, then we can force the agents along this path to deviate and increase the allocation of $i$ without changing the total allocation, cost and utility of other agents. This is a contradiction. Now assume that an augmenting path from $i$ does not exist. Let $S$ be the set of agents and $T$ be the set of bundles which are reachable from $i$. Since all the bundles in $T$ have been fully consumed by agents in $S$, and each agent in $S$ has at most 1 unit of bundles in $T$, we have
$$\sum_{b \in B_S} x_b < |S|$$
This is a contradiction, and we are done.
\end{proof}

The problem with lemma \ref{Lem:GEChar} is that in order to check whether the first set of conditions hold, for every subset of agents we have to define a polynomial and there are exponentially many of these subsets. For a subset of bundles $T$, let $A(T)$ be the set of agents that for every agent $i$, $i \in A(T)$ if and only if $B_i \subseteq T$. In order to fix the problem, we prove the following lemma.

\begin{lemma}\label{Lem:GEEquiv}
For every subset of agents $S$, we have $\sum_{b \in B_S} x_b \geq |S|$ if and only if for every subset of bundles $T$, $\sum_{b \in T} x_b \geq |A(T)|$.
\end{lemma}
\begin{proof}

$\Rightarrow$ Let $S=A(T)$. By definition we know that $B_S\subseteq T$. So we have
$$|A(T)| = |S| \leq \sum_{b \in B_S} \Allocs_b \leq \sum_{b \in T} \Allocs_b$$

$\Leftarrow$ Let $T=B_S$. By definition we know that $S\subseteq A(T)$. So we have
$$|S| \leq |A(T)| \leq \sum_{b \in T} x_b = \sum_{b \in B_S} \Allocs_b$$
\end{proof}

Note that since the number of goods is constant, the number of subsets of bundles is also constant. Lemma \ref{Lem:GEEquiv} shows that when we know the bundles and optimum bundles of all agents, we can find out whether the the first condition of lemma \ref{Lem:GEChar} holds by defining a variable for each bundle and checking polynomial number of inequalities. It is also easy to see that we can check the second set of conditions by defining a polynomial (whose degree is a function of number of bundles) for each agent, using the formula from the previous section for each $\alpha_b$.

We summarize the arguments in this section in the following lemma
\begin{lemma}\label{Lem:SecGEqChar}
For each output cell of theorem \ref{THM:CellDecomp}, defined by variables and polynomials in the previous section and above $O(m^2)$ variables and $O(2^{m^2})$ polynomials, we can use the sign of the defined polynomials to see if the prices in that cell are equilibrium prices.
\end{lemma}
\begin{proof}
The proof of this lemma is similar to proof of the Hall's theorem. On one side of the bipartite graph we have agents with capacity 1 and on the other side we have the bundles with capacity $\Allocs_b$ (for each bundle $b\in B$). Furthermore, there is an edge between bundle $b$ and agent $i$ iff $b \in B_i$.

$\Rightarrow$ This case directly follows from lemma \ref{Lem:Bundling} and the market clearing conditions.

$\Leftarrow$ Assume that there is an $\Allocss\in \PosReals^{|B|}$ such that a and b hold. We have to show that there exists an allocation of bundles to agents such that every agent gets exactly 1 unit of her optimum bundles. We use contradiction. Assume that this allocation does not exists. Let $Y$ be an allocation which allocates maximum total amount of bundles to agents among all the valid allocations. Let $i$ be an agent for which $Y$ allocates less than 1 unit of her optimum bundles. If there is an augmenting path from $i$ to a bundle $b$ which is not fully consumed, then we can force the agents along this path to deviate and increase the allocation of $i$ without changing the total allocation, cost and utility of other agents. This is a contradiction. Now assume that an augmenting path from $i$ does not exist. Let $S$ be the set of agents and $T$ be the set of bundles which are reachable from $i$. Since all the bundles in $T$ have been fully consumed by agents in $S$, and each agent in $S$ has at most 1 unit of bundles in $T$, we have
$$\sum_{b \in B_S} \Allocs_b < |S|$$
This is a contradiction, and we are done.
\end{proof}

\subsection{Finding an Equilibrium}\label{Sec:GoodsFindingEq}

So far, we have characterized the equilibria by defining constant number of variables and polynomial number of polynomials. Now we put the pieces together to prove theorem \ref{THM:GoodsMain}.

\textbf{Proof of Theorem \ref{THM:GoodsMain}. } We use theorem \ref{THM:CellDecomp} with variables and polynomials defined in section \ref{Sec:GoodsBundleChar} and \ref{Sec:GoodsEqChar}. Theorem \ref{THM:GoodsMain} iterates all the feasible cells of these polynomials. Each cell gives us the sign of all these polynomials. The signs of these polynomials tell us
\begin{enumerate}
\item What are the bundles and optimum bundles for each agent, by lemma \ref{Lem:SecGBun}.
\item Whether these bundles are associated with a market equilibrium, by lemma \ref{Lem:SecGEqChar}.
\end{enumerate}

Finally, if a bundles associated with a cell is characterizing an equilibrium bundling, we sample a set of prices from that cell. Now that we know equilibrium prices, it is easy to find the bundles as before. We can also find each agent's allocation by finding a solution of the following set of inequalities

\begin{align*}
&\sum_{b \in B_i}\Alloc_{ib}=1 &\forall i\in[n]\\
&\sum_{i}\big(\sum_{b:b=(j,k)}\alpha_b\Alloc_{ib} + \sum_{b:b=(k,j)}(1-\alpha_b)\Alloc_{ib} + \sum_{b:b=(j)} \Alloc_{ib} \big) =C_j &\forall j\in[m]\\
&\Alloc_{ib} \geq 0 &\forall i\in[n], b\in B\\
\end{align*}

In which, $\Alloc_{bi}$ is the amount of bundle $b$ used by agent $i$. Note that since we know this bundling is associated with an equilibrium, the feasible region of the above inequalities is not empty. Finally, the allocation of each agent $i$ for item $j$ in this equilibrium is
$$\Alloc_{ij}=\sum_{b:b=(j,k)}\alpha_b\Alloc_{ib} + \sum_{b:b=(k,j)}(1-\alpha_b)\Alloc_{ib} + \sum_{b:b=(j)} \Alloc_{ib}$$

Note that since theorem \ref{THM:CellDecomp} iterates over all the possible feasible cells, if the equilibrium exists, our method finds it. Furthermore, our method can characterize all the equilibrium cells.

Since the number of variables is constant ($O(m^2)$), the number of polynomials is polynomial in the number of agents ($O(2^{m^2} + nm^4)$) and all the polynomials that we define have constant degree ($O(m^2)$), our algorithm runs in polynomial time.

\section{Examples}
\subsection{Examples from the Preliminaries} \label{ap:examples}
\paragraph{Interim Pareto-efficiency of RSD}
It is well known that \RSD{} is ex post Pareto-efficient \etedit{(efficient after the random choices are made), but not interim Pareto-efficient. From the perspective of a fixed agent, ordinal preferences are not always sufficient for ranking the interim outcomes of a mechanism
(i.e., ranking of distributions over outcomes), even though they are sufficient for ranking the ex-post outcomes. However, in the example below, ordinal preferences are enough to show that the interim outcomes are not Pareto-efficient. Consider} the following example with agents $\AgentID{1},\ldots,\AgentID{4}$ and items
$\ItemID{1},\ldots,\ItemID{4}$ in which the preferences of the agents are
indicated on the edges (only the first two top choices of each agent).
\begin{center}
\begin{tikzpicture}[marketstyle]
    \def\marketsize{4};
    \foreach \i in {1,...,\marketsize}
        \node["\AgentID{\i}"below,circle,fill] (agent_\i) at (\i,0) {};
    \foreach \j in {1,...,\marketsize}
        \node["\ItemID{\j}"above,circle,fill] (item_\j) at (\j,1) {};

    \draw(agent_1)
        edge["1st"] (item_1)
        edge["2nd"] (item_2);
    \draw(agent_2)
        edge["2nd"] (item_1)
        edge["1st"] (item_2);
    \draw(agent_3)
        edge["1st"] (item_1)
        edge["2nd"] (item_2);
    \draw(agent_4)
        edge["2nd"] (item_1)
        edge["1st"] (item_2);
\end{tikzpicture}
\end{center}
Let $\Alloc_{ij}$ denote the probability that agent $i$ gets item $j$ assuming
the items are allocation using \RSD. Notice that \AgentID{1} gets \ItemID{2}
when the agents are served in the order
$\AgentID{3},\AgentID{1},\AgentID{2},\AgentID{4}$. Similarly $\AgentID{2}$ gets
$\ItemID{1}$ when the agents are served in the order
$\AgentID{4},\AgentID{2},\AgentID{1},\AgentID{3}$. Therefore both $\Alloc_{12}$
and $\Alloc_{21}$ are strictly positive. However that implies the allocation is
not interim Pareto-efficient because \AgentID{1} and \AgentID{2} would both
strictly benefit by exchanging some fraction of their allocation of \ItemID{1}
and \ItemID{2}.

\paragraph{Pareto-efficiency with cardinal versus ordinal preferences}
\etedit{Cardinal and ordinal preferences are different in evaluating outcomes. In the next} example 
the \RSD{} interim allocations is interim Pareto inefficient with cardinal preferences, even though with the corresponding ordinal preferences would not have been interim Pareto ifefficient. Consider the following example in which the valuation of each agent for
each item is specified on the corresponding edge and $\epsilon < 0.5$.
\begin{center}
\begin{tikzpicture}[marketstyle]
    \def\marketsize{3};
    \foreach \i in {1,...,\marketsize}
        \node["\AgentID{\i}"below,circle,fill] (agent_\i) at (\i,0) {};
    \foreach \j in {1,...,\marketsize}
        \node["\ItemID{\j}"above,circle,fill] (item_\j) at (\j,1) {};
    \draw(agent_1)
        edge["$1$"] (item_1)
        edge["$\epsilon$"] (item_2)
        edge["$0$"] (item_3);
    \draw(agent_2)
        edge["$1$"] (item_1)
        edge["$1-\epsilon$"] (item_2)
        edge["$0$"] (item_3);
    \draw(agent_3)
        edge["$1$"] (item_1)
        edge["$1-\epsilon$"] (item_2)
        edge["$0$"] (item_3);
\end{tikzpicture}
\end{center}
All of the agents in the above example have the same ordinal preferences, but
different cardinal preferences. Running the $\RSD$ mechanism would yield an
allocation of $\Alloc_i=(\frac{1}{3},\frac{1}{3},\frac{1}{3})$ for each agent $i
\in \RangeN{3}$ which would have been Pareto-efficient for the corresponding
ordinal preferences. However this allocation is not Pareto efficient for the
above cardinal preferences because it is strictly Pareto dominated by the
following allocation for any $\lambda \in \RangeAB{0}{1}$:
\begin{align*}
    \Allocs_1
        &= \bigl(\frac{1}{3}+\frac{\lambda(1-\epsilon)+(1-\lambda)\epsilon}{3}, 0, \frac{1}{3}+\frac{\lambda\epsilon+(1-\lambda)(1-\epsilon)}{3}\bigr), \\
    \Allocs_2 = \Allocs_3
        &= \bigl(\frac{1}{3}-\frac{\lambda(1-\epsilon)+(1-\lambda)\epsilon}{6}, \frac{1}{2}, \frac{1}{3}-\frac{\lambda\epsilon+(1-\lambda)(1-\epsilon)}{3}\bigr).
\end{align*}
Basically the interim \RSD{} allocation can be Pareto improved by agent 1
trading its whole share of item 2 with agent 2 and 3 in return for a smaller
share of item 1 where item 3 is just traded as a \emph{dummy item} to allow the
total allocation of each agent to remain equal to $1$.



\paragraph{Odd Demand Structure}
As mentioned in the Introduction, the agents demand in the matching problem exhibits
some peculiar properties, that make the market equilibrium problem challenging to solve.
The demand function
does not satisfy the gross substitutability condition: increasing the price of
an item (while keeping all other prices fixed) may actually increase the demand
of the agent for that item. Here we provide an example of this phenomenon. Let $\Val_{ij} \in \SPosReals$ denote the valuation
of agent $i \in \RangeN{\NumAgents}$ for item $j \in \RangeN{\NumItems}$. Given
a price vector $\Prices \in \PosReals^\NumItems$, the optimal demand bundles of
agent $i$ are the solutions of the following linear program whose objective
value is the optimal expected utility of the agent:
\begin{align}
    &\text{maximize}   & \sum_j & \Val_{ij} \Alloc_{ij} \label[program]{lp:util} \\
    &\text{subject to}  & \sum_j \Alloc_{ij} \Price_j & \le 1 \tag{$\alpha_i$} \\
                    &   & \sum_j \Alloc_{ij} & \le 1 \tag{$\beta_i$} \\
                    &   & \Alloc_{ij} &\ge 0 & &\forall j \in \RangeN{\NumItems}. \notag
\end{align}
Let $\alpha_i \in \PosReals$ and $\beta_i \in \Reals$ be the dual variables
corresponding to the first and second constraint respectively. The dual of the
above program is as follows:
\begin{align}
    &\text{minimize}   & \alpha_i &+ \beta_i \\
    &\text{subject to}  & \alpha_i \Price_j + \beta_i &\ge \Val_{ij}, & &\forall j \in \RangeN{\NumItems}  \tag{$\Alloc_{ij}$}\\
                    &   & \alpha_i, \beta_i &\ge 0. \notag
\end{align}
It is easy to see that $\beta_i = \max_j (\Val_{ij}-\alpha_i\Price_j)$ in any
optimal assignment. By eliminating $\beta_i$ from the above program and
rearranging the terms we can derive the optimal utility of agent $i$ as a
function of the prices as
\begin{align}
    \Util_i(\Prices) &= \min_{\alpha_i \ge 0} \DualUtil_i(\alpha_i,\Prices), \label{eq:util:min} \\
\intertext{in which}
        \DualUtil_i(\alpha_i,\Prices) &= \alpha_i+{\max_j}^+ \left(\Val_{ij}-\alpha_i\Price_j\right).
        \label{eq:util:max}
\end{align}
The $\max^+$ operator is taken over linear functions of $\alpha_i$, therefore
$\DualUtil_i(\alpha_i,\Prices)$ is piecewise linear and convex in $\alpha_i$ for
every fixed $\Prices$. Consequently $\Util_i(\Prices)$ can be computed easily
using a linear search over $\alpha_i$ by checking only the endpoints of the
linear pieces.

Next, we compute the optimal fractional demand bundle $\Allocs_i$ of agent $i$
at price vector $\Prices$. Let $\DemandSup_i(\Prices)$ be the set of all items
$j$ which are maximizers of \cref{eq:util:max} in which $\alpha_i$ is a
minimizer of \cref{eq:util:min}. Assuming $\alpha_i$ is non-zero, the minimum of
$\DualUtil_i(\alpha_i, \Price)$ takes place either at the intersection of a
linear piece with negative slope and a linear piece with a positive slope, or on
a linear piece with a zero slope. The slope of the line segment corresponding to
item $j$ is $1-\Price_j$. Therefore, there must exist
$j',j'' \in \DemandSup_i(\Prices)$ such that
$\Price_{j'} \ge 1 \ge \Price_{j''}$. ~\footnote{Note that $j'$ and $j''$ might
  be the same if $\Price_{j'}=1$.} In general $\DemandSup_i(\Prices)$ may
include more than two items, however in the most common case we have
$\DemandSup_i(\Prices)=\{j',j''\}$ where $\Price_{j'} > 1 >
\Price_{j''}$. Notice that $\Alloc_{ij}$ can be non-zero only if
$j \in \DemandSup_i(\Prices)$. Furthermore complementary slackness implies that
if $\alpha_i > 0$, then the budget must be completely spent. So we have two
possibly non-zero variables (i.e., $\Alloc_{ij'},\Alloc_{ij''}$) and two tight
constraints (i.e., the first and second constraints in~\cref{lp:util} which can
be solved to get
\begin{align*}
    \Alloc_{ij'}    &= \frac{1 - \Price_{j''}}{\Price_{j'}-\Price_{j''}}, \\
    \Alloc_{ij''}   &= \frac{\Price_{j'}-1}{\Price_{j'}-\Price_{j''}}.
\end{align*}

Next we show that the agent's utility function does not satisfy the \emph{weak
  gross substitutability} condition. Suppose $\Price_{j''}$ is increased by a
positive amount which is small enough not to cause $\DemandSup_i(\Prices)$ to
change. The derivative of $\Alloc_{ij'}$ with respect to $\Price_{j''}$ is
\begin{align*}
  \od{\Alloc_{ij'}}{\Price_{j''}} &= \frac{1-\Price_{j'}}
                                    {(\Price_{j'}-\Price_{j''})^2} < 1
\end{align*}
which means increasing the price of item $j''$ would decrease the agent's demand
of item $j'$ which violates the gross substitutability
requirement. Surprisingly, that also implies increasing the price of item $j''$
must increase the demand for item $j''$ itself because
$\Alloc_{ij'}+\Alloc_{ij''}=1$.


\section{Existence of Equilibria}
\label{sec:existence}
In this \lcnamecref{sec:existence} we prove the existence of a market
equilibrium by constructing a corresponding concave game with convex externality
constraints as formally defined in \cref{sec:game} such that any equilibrium of
the concave game corresponds to a market equilibrium. We then show that this
game satisfies the requirements of \cref{res:ngame} which implies it has an
equilibrium, hence a market equilibrium exists.

Given a market with $\NumAgents$ unit demand agents and $\NumItems$ items with
$\ItemCount_j$ copies for each item $j \in \RangeN{\NumItems}$, we define an
$\NumAgents+1$ player concave game as follows. For each
$i \in \RangeN{\NumItems}$, player $i$ corresponds to agent $i$ in the market
and chooses $\Allocs_i \in \PosReals^\NumItems$ to optimize the following linear
program:
\begin{align}
    &\text{maximize}   & \sum_j & \Val_{ij} \Alloc_{ij} \label[program]{lp:util2} \\
    &\text{subject to}  & \sum_j \Alloc_{ij} \Price_j & \le 1 \notag \\
                    &   & \sum_j \Alloc_{ij} &\le 1 \notag  \\ 
                    &   & \Alloc_{ij} &\ge 0 & &\forall j \in \RangeN{\NumItems}. \notag
\end{align}
Player $0$ corresponds the market maker and chooses the price vector
$\Prices \in \PosReals^\NumItems$ to optimize the following linear program:
\begin{align}
    &\text{maximize}   & \sum_j &(\sum_i \Alloc_{ij} - \ItemCount_j) \Price_j \label[program]{lp:price} \\
    &\text{subject to}  & \sum_j \ItemCount_j \Price_j &\le \NumAgents \notag \\
                    &   & \Price_j &\ge 0 & &\forall j \in \RangeN{\NumItems}. \notag
\end{align}

\begin{lemma}
\label{res:reduction}
Any equilibrium of the above game corresponds to an equilibrium of the original market.
\end{lemma}
\begin{proof}
  It is easy to see that in any equilibrium of the game the action of player
  $i$, $\Allocs_i$, is an optimal allocation for agent $i$ of the original
  market with respect to price vector $\Prices$.  So we only need to verify that
  the market clearing conditions are also satisfied. It is easy to see that if
  there are over-allocated items, then player $0$ (market maker) should have
  chosen a price vector $\Prices$ that would make the object value of
  \cref{lp:price} strictly positive and would also make the first constraint
  tight. But a strictly positive objective means the total amount of money spend
  by all agents, $\sum_j (\sum_i \Alloc_{ij} \Price_j)$, is more than
  $\sum_j (\sum_i \ItemCount_j \Price_j)$ which is itself equal to $\NumAgents$
  by tightness of the first constraint, hence a contradiction.  On the other
  hand, if there is any under-allocated item, the market maker should have
  assigned a price of $0$ to that item, hence all agents must be fully
  allocated.
\end{proof}

\begin{theorem}
\label{res:existence}
A market equilibrium for the original market always exists.
\end{theorem}
\begin{proof}
  By \cref{res:reduction} an equilibrium for the original market exists if the
  corresponding concave game has an equilibrium. By \cref{res:ngame} a concave
  game has an equilibrium if every player has a default action that lies
  strictly inside the feasible set of actions for that player for all possible
  actions of other players. Let $\epsilon = 1/(\NumAgents\sum_j
  \ItemCount_j)$. A default action for player $0$ is given by
  $\Price^0_j = \epsilon$ for every $j\in \RangeN{\NumItems}$. A default action
  for each player $i\in \RangeN{\NumAgents}$ is given by
  $\Allocs^0_{ij} = \epsilon$ for all $j\in \RangeN{\NumItems}$.
\end{proof}

\subsection{Multi Agent Concave Games with Externality Constraints}
\label{sec:game}
In this \lcnamecref{sec:game} we prove more generally the existence of an equilibrium for a
general class of games in which both the utility function and the set of
feasible actions for a player may depend on the actions of the other
players. Formally, suppose there are $\NumAgents$ players and the optimal utility
of player $i \in \RangeN{\NumAgents}$ is captured by the following convex program
\begin{align}
    &\text{maximize}   & \Val_i(\Allocs_i,\Allocss_{-i}) \label[program]{cp:util} \\
    &\text{subject to}  & \ConsFun_{ik}(\Allocs_i,\Allocss_{-i}) &\le 0 & &\forall k \in \RangeN{\NumCons_i} & \notag \\
                    &   & \Allocs_i & \in \AllocSpace_i \subseteq \Reals^{\AgentDim_i} \notag
\end{align}
in which $\Allocs_i$ is the action of player $i$; $\Allocss_{-i}$ is the vector
of actions of players other than $i$; $\Val_i$ is the utility function of player
$i$ which is concave in $\Allocs_i$ and continuous in all arguments;
$\AllocSpace_i \subset \Reals^{\NumItems_i}$ is a compact and convex set
representing the feasible actions of player $i$; and $\ConsFun_{ik}$ is convex in
$\Allocs_i$ and continuous in all arguments. Define
$\AllocSpaces = \AllocSpace_1\times\cdots\times\AllocSpace_\N$. A vector of
actions $\Allocss=(\Allocs_1,\ldots,\Allocs_\N) \in \AllocSpaces$ is an
\emph{equilibrium} if{f} $\Allocs_i$ is an optimal assignment for \cref{cp:util}
for each $i \in \RangeN{\N}$. The next theorem establishes that an equilibrium
always exists under a mild assumption.

\begin{theorem}
\label{res:ngame}
If for each player $i \in \RangeN{\NumItems}$ there exists a \emph{default
  action} $\Allocs_i^0 \in \AllocSpace_i$ for which the first set of constraints
in \cref{cp:util} hold strictly for all $\Allocss_{-i} \in \AllocSpaces_{-i}$,
then an equilibrium always exists.
\end{theorem}

\begin{proof}
Define the set of optimal actions of player $i$ in response to
$\Allocs_{-i} \in \AllocSpaces_{-i}$ as
\begin{align*}
  \Response_i(\Allocss_{-i}) &= \left\{\Allocs_i \in \AllocSpace_i \suchthat \text{$\Allocs_i$ is a solution to \cref{cp:util} with $\Allocss_{-i}$} \right\}.
\end{align*}
Define the collective optimal actions in response to $\AllocsY \in \AllocSpaces$ as
\begin{align*}
  \Response(\AllocssY) &= \left\{ \Allocs \in \AllocSpaces \suchthat \Allocs_i \in \Response_i(\AllocssY_{-i}) \quad \forall i \in \RangeN{\N} \right\}.
\end{align*}
It is easy to see that $\Response$ is a set valued function from $\AllocSpaces$
to $2^\AllocSpaces$ whose fixed points correspond to the equilibria of the
game. We prove that $\Response$ satisfies the requirements of Kakutani's fixed
point theorem and therefore must have a fixed point which would then imply the
existence of an equilibrium. It is easy to see that $\AllocSpaces$ is a
non-empty compact and convex set, and that $\Response(\AllocssY)$ is a also a
non-empty compact convex set for every $\AllocssY \in \AllocSpaces$. Therefore
we only need to prove that $\Response$ has a closed graph. Formally, we need to
show that for any sequence $({\Allocss}^t,{\AllocssY}^t)_{t\in \PosInts}$ where
${\Allocss}^t \in \Response(\AllocssY^t)$ for all $t \in \PosInts$, if
$\lim_{t \to \infty} (\Allocss^t,\AllocssY^t) = (\Allocss^*,\AllocssY^*)$, then
$\Allocss^* \in \Response(\AllocssY^*)$. Consider \cref{cp:util} in which
$\Allocss_{-i}$ set to $\AllocssY^*_{-i}$. From the continuity of
$\ConsFun_{ik}$ and compactness of $\AllocSpace_i$, it follows that
$\Allocss^*_i$ is a feasible action for player $i$ in response to
$\AllocssY^*_{-i}$, so we only need to show that it is also optimal. Pick any
$\Allocs^\dagger_i \in \Response_i(\AllocssY^*_{-i})$. We will prove optimality
by showing that
$\Val_i(\Allocs_i^*,\AllocssY^*_{-i}) \ge
\Val_i(\Allocs^\dagger_i,\AllocssY^*_{-i})$. Recall that
$\Val_i(\Allocs_i^*,\AllocssY^*_{-i})=\lim_{t\to
  \infty}\Val_i(\Allocs_i^t,\AllocssY^t_{-i})$, so ideally we would like to show
that
$\Val_i(\Allocs_i^t,\AllocssY^t_{-i}) \ge
\Val_i(\Allocs_i^\dagger,\AllocssY^t_{-i})$, however $\Allocss_i^\dagger$ may
not even be a feasible action in response to $\AllocsY_{-i}^t$ for any $t$ and
therefore such an inequality may not hold in general.
This is where we use the assumption of the theorem to show that there exists an
action ${\Allocss_i'}^t$ between $\Allocs_i^t$ and $\Allocs_i^\dagger$ which is
feasible for player $i$ in response to $\AllocssY_{-i}^t$ for all $t$ and also
converges to $\Allocs_i^\dagger$ as $t\to \infty$ which will imply the desired
inequality as follows:
\begin{align*}
  \Val_i(\Allocs_i^*,\AllocssY^*_{-i})
     &= \lim_{t\to \infty}\Val_i(\Allocs_i^t,\AllocssY^t_{-i}) \\
     &\ge \lim_{t\to \infty}\Val_i({\Allocs'_i}^t,\AllocssY^t_{-i}) \\
     &= \Val_i(\Allocs_i^\dagger,\AllocssY^*_{-i}).
\end{align*}
To define ${\Allocs'_i}^t$, we first define the following two quantities in
which $\Allocs^0_i$ is the default action of player $i$:
\begin{align*}
    \delta^t &= \max\left(0, \max_{i, k} \ConsFun_{ik}(\Allocs_i^\dagger,\AllocssY^t_{-i})\right)\\
    \delta^0 &= -\max_{i, k, \Allocss''_{-i}} \ConsFun_{ik}(\Allocs^0_i,\Allocss''_{-i}).
\end{align*}
Note that $\delta^0 > 0$ because of the hypothesis of the theorem. Let
$\lambda^t = \frac{\delta^0}{\delta^0+\delta^t}$. Define
${\Allocss'}^t = \lambda^t \Allocss^\dagger+(1-\lambda^t) \Allocss^0$. First
notice that ${\Allocs'_i}^t$ is a feasible response to $\AllocssY^t_{-i}$
because
\begin{align*}
    \ConsFun_{ik}({\Allocss'}_i^t, \AllocssY^t_{-i})
        &\le \lambda^t \ConsFun_{ik}(\Allocss^\dagger_i, \AllocssY^t_{-i})+(1-\lambda^t)\ConsFun_{ik}(\Allocss^0_i, \AllocssY^t_{-i}) & &\text{by convexity of $\ConsFun_{ik}$ in ${\Allocss'}_i^t$} \\
        &\le \lambda^t\delta^t-(1-\lambda^t)\delta^0
            & &\text{by definition of $\delta^t$ and $\delta^0$} \\
        &= 0 & &\text{by definition of $\lambda^t$}.
\end{align*}
Second, notice that continuity of $\ConsFun_{ik}$ guarantees
$\lim_{t \to \infty} \delta^t = 0$, which implies
$\lim_{t\to\infty} \lambda^t=1$ (recall that $\delta^0 > 0$), which implies
$\lim_{t \to \infty} {\Allocss'}^t = \Allocss^\dagger$ as claimed.


\end{proof}

%


\end{document}